\documentclass[conference]{IEEEtran}  
\usepackage{amsmath,graphicx}
\usepackage{xcolor}

\usepackage{tikz}
\usetikzlibrary{positioning}
\usetikzlibrary{arrows.meta}
\usetikzlibrary{calc}
\usetikzlibrary{shapes.geometric,shapes.misc}
\newtheorem{definition}{\textbf{Definition}}
\newtheorem{lemma}{\textbf{Lemma}}
\newtheorem{theorem}{\textbf{Theorem}}

\newtheorem{remark}{\textbf{Remark}}

\usepackage{amssymb}
\usepackage{algorithm}
\usepackage{algpseudocode}
\usepackage{amsmath, amssymb}
\usepackage{tikz}
\usetikzlibrary{positioning, arrows.meta, shapes, calc, backgrounds, fit}
\usetikzlibrary{positioning,arrows.meta,decorations.pathreplacing,shapes}
\usepackage{pgfplots}
\pgfplotsset{compat=1.18}
\usepackage{amsmath}
\usepackage{graphicx}

\usepackage{amsmath, amssymb, amsfonts}  
\usepackage{graphicx}                    
\usepackage{tikz}                        
\usepackage{cite}                         
\usepackage{url}                          
\usepackage{bm}                           
\usepackage{multirow}                     
\usetikzlibrary{backgrounds}
\usepackage{pgfplots}

\def\BibTeX{{\rm B\kern-.05em{\sc i\kern-.025em b}\kern-.08em
    T\kern-.1667em\lower.7ex\hbox{E}\kern-.125emX}}
\begin{document}

\title{RankGuard-Polar: Private–Public Finite Length Polar Codes with Rank-Certified Leakage $^{\star}$}

\author{\IEEEauthorblockN{Hassan Tavakoli}
\IEEEauthorblockA{\textit{School of EECS} \\
\textit{ Oregon State University}\\
Oregon, OR, 97331 USA \\
tavakolh@oregonstate.edu}
\and
\IEEEauthorblockN{Thinh Nguyen, \emph{Senior Member, IEEE}}
\IEEEauthorblockA{\textit{School of EECS} \\
\textit{ Oregon State University}\\
Oregon, OR, 97331 USA \\
thinhq@eecs.oregonstate.edu}
\and
\IEEEauthorblockN{Bella Bose, \thanks{$^{\star}$ This work was supported by the National Science Foundation under Grant No. CCF-2417898.} \emph{Life Fellow, IEEE}}
\IEEEauthorblockA{\textit{School of EECS} \\
\textit{ Oregon State University}\\
Oregon, OR, 97331 USA \\
Bella.Bose@oregonstate.edu}
}

\maketitle

\begingroup
\renewcommand\thefootnote{}
\footnotetext{$^{\star}$ This work was supported by the National Science Foundation under Grant No. CCF:SHF:2417898.}
\endgroup

\begin{abstract}
We introduce \textbf{RankGuard-Polar}, a framework for safely publishing a subset of polar codeword coordinates over shared public resources. We assume a strong eavesdropper who has access to the channel input, i.e., the transmitted codeword coordinates published on a public resource access model. Working over \(\mathbb F_2\) and focusing on time-shared public/private BEC uses, we show that leakage from a published index set \(\mathbf{P}\) admits an exact algebraic characterization comes from an information-theoretic viewpoint, and we construct an explicit linear extractor ($R$) that identifies the leaked linear combinations. Building on this identity, we (i) give efficient procedures to compute and certify leakage for any \(\mathbf{P}\), (ii) propose a practical fast algorithm with provable efficiency.
\end{abstract}

\begin{IEEEkeywords}
Polar codes, Public Channel, Private Channel, frozen bits, mutual information, rank methods
\end{IEEEkeywords}


\section{Introduction and Overview}
\label{sec:intro}

Channel coding combats channel noise by adding structured redundancy so the receiver can correct errors and recover the original message reliably \cite{shannon1948,cover2006,lin2004}.  That redundancy, however, increases the number of transmitted symbols and thus enlarges the attack surface available to an adversary who can observe or tamper with some transmissions.  For this reason the interplay between coding and cryptography has been studied extensively: redundancy can aid legitimate decoding while, if combined with secret structure or randomness, hindering an attacker.  An example is the McEliece cryptosystem, which uses joint error-correcting code with secret linear transformations so legitimate receivers can decode efficiently while attackers face a computationally hard decoding problem \cite{mceliece1978,niederreiter1986}.
Previous works have shown that polar codes can be combined with other communication settings and design goals. In \cite{tavakoli2018polarization,tavakoli_good_index_polarized_relay}, the authors studied polarized multi-relay channels and good index selection for relay settings, and also the author of \cite{tavakoli2017polarization}, later extended polarization ideas to a multiple-access-based construction. Motivated by this line of work, we study polar codes with rank-based leakage certification in a finite-length public/private framework.

Secrecy can be embedded into the encoder--decoder pair, but using the generator matrix \(G\) itself as a secret key is restrictive since \(G\) must preserve reliability and decoding efficiency. Practical code-based cryptosystems (e.g., McEliece) therefore hide a structured code via secret scramblers and permutations rather than choosing an arbitrary generator matrix.
In many practical deployments, a transmitter can use shared public resources to increase throughput or reduce cost. This raises a simple operational question:
\emph{Under the strong assumption that the eavesdropper has access to the input of the public channel (i.e., the eavesdropper observes the channel input), we ask: How much information about the confidential message leaks when we place some codeword coordinates on that public resource?
}  

In this paper, we exploit a practical polar-code mechanism for information-theoretic masking: the polar transform \(G_N\) is fixed, and the encoder already designates a set of positions as \emph{frozen}. By drawing frozen bits from fresh secret randomness, one can mask published coordinates without changing the transform or code dimension. However, publish/no-publish decisions must be certified so that no, or as few as possible, independent linear combinations of the information bits are exposed. We address this finite-blocklength problem by developing exact algebraic leakage certificates for any public index set \(\mathbf{P}\), together with constructive extractors and practical selection algorithms. Our focus is on polar-coded systems over stationary erasure channels, in particular the BEC.

The rest of the paper is organized as follows.
In Section~\ref{sec:system_model},  we introduce our motivation and preliminaries. 
Section~\ref{sec:RankGuard} develops \textit{RankGuard-Polar}: we prove the exact leakage identity.
Section~\ref{sec:Algorithms} describes practical algorithms for selecting public indices and shows how we can reduce the complexity.
In Section~\ref{sec:Relations}, we compare our finite-block, public/private viewpoint to classical wiretap and general cryptography problem.
Finally, Section~\ref{sec:Conclusion} summarizes contributions and outlines directions for extensions.
\section{Motivation and Preliminaries} \label{sec:system_model}

\subsection{Definitions:}

We adopt the following notational conventions throughout the paper. All linear algebra is over the binary field $\mathbb F_2$ unless stated otherwise. Vectors, bold letters like $\mathbf{u}$ , and matrices, capital letters like $G$, are representing. Shannon entropy and mutual information are $H$ and $I$, respectively, in bits and base~2. $e_i$ is the $i$-th standard basis column vector. For a vector $\mathbf{v}$, $\|\mathbf{v}\|_0$ denotes Hamming weight, also, $diag(\mathbf{v})$ is the $n\times n$ diagonal matrix with $\mathbf{v}$ on the diagonal. For a finite set $S$, $|S|$ is its cardinality. We also write $[1:N] \triangleq \{1,\dots,N\}$, and \(u^N \triangleq (u_1,\dots,u_N)\). 
We also adopt the standard polar-code notation introduced by \cite{Arikan2009}.For blocklength \(N = 2^n\), the polar transform is given by
\(
G_N = G_2^{\otimes n},
\)
where
\(
G_2 =
\begin{bmatrix}
1 & 0 \\
1 & 1
\end{bmatrix},
\)
and \(\otimes\) denotes the Kronecker product. We also assume for index sets $\mathbf{A}$ and matrix \(G\), \(G_{\mathbf{A}}\equiv (G_N)_{\mathbf{A}}\).

Let \( \mathbf{A}\subseteq\{1,\dots,N\}\) be the information set and \(\mathbf{ F}=\{1,\dots,N\}\setminus  \mathbf{A}\) the frozen set. Let \(\mathbf{P}\subseteq\{1,\dots,N\}\) denote a subset of output positions of polar encoder.

\begin{definition}[Field and vectors]
We write \(\mathbf{u}=(\mathbf{u_A},\mathbf{u_F})\) where \(\mathbf{u_A}\in\mathbb F_2^{1\times |A|}\) are the information bits and \(\mathbf{u_F}\in\mathbb F_2^{1\times|F|}\) are the frozen bits. For a set of published indices, public indices, \(\mathbf{P}\) we write \(\mathbf{x_p}\in\mathbb F_2^{1\times|P|}\) for the corresponding subset of codeword coordinates.
\end{definition}

\begin{definition}[Submatrices]
For a matrix \(G\in\mathbb F_2^{N\times N}\) and index set \(\mathbf{P}\subset[1:N]\) let
\(
G_{\mathbf{A,P}}\in\mathbb F_2^{|A|\times|P|},\quad\text{and} \; G_{\mathbf{F,P}}\in\mathbb F_2^{|F|\times|P|},
\)
denote the submatrices of \(G\) formed by restricting rows to \(\mathbf{P}\) public index and columns to the information index set and the frozen index set, respectively. Thus an adversary observing only public bits as: 
\begin{align}
\mathbf{x_p} = \mathbf{u_A} G_{\mathbf{A,P}} \oplus \mathbf{u_F} G_{\mathbf{F,P}}\in\mathbb F_2^{1\times|P|}.
\label{eq:xp}
\end{align}
We  abbreviate \(G_{\mathbf{A,P}}=G_{P,\mathrm{info}}\) and \(G_{\mathbf{F,P}}=G_{P,\mathrm{frozen}}\). We also assume that \(\mathbf{u_A}\in\mathbb F_2^{1\times|A|}, \mathbf{u_F}\in\mathbb F_2^{1\times|F|},
\)
\end{definition}

\begin{definition}[Image]
For a matrix \(M\in\mathbb F_2^{m\times n}\) the image is
\(
\mathrm{Im}(M) \;=\; \{\, \mathbf{v}M : \mathbf{v}\in\mathbb F_2^m \,\}\subseteq\mathbb F_2^n.
\)
If \(\operatorname{rank}(M)=r \leq n\) then \(|\mathrm{Im}(M)|=2^r\).
\end{definition}

\begin{definition}
Let \(W:\{0,1\}\to\mathcal{Y}\) be a binary-input memoryless channel and let \(G_N\) denote the length-\(N=2^n\) polar transform. Define the channel input \(x^N=u^N G_N\) and the product channel \(W^N(y^N\mid x^N)=\prod_{j=1}^N W(y_j\mid x_j)\). The \(i\)-th synthetic channel \(W_N^{(i)}\) is defined as
\(
W_N^{(i)}\big(y^N,u^{1:i-1}\mid u_i\big)
\triangleq\;
\frac{1}{2^{\,N-1}}
\sum_{u_{i+1}^N\in\{0,1\}^{N-i}}
W^N\big(y^N\mid x^N\big).
\)
Thus, \(W_N^{(i)}\) represents the effective channel seen by bit \(u_i\) when the
decoder has access to the channel output \(y^N\) and the previously decoded bits
\(u^{1:i-1}\), while the remaining bits \(u_{i+1}^N\) are averaged out and treated
as uniformly random.
\end{definition}

\begin{definition}
Let \(W_N^{(i)}\) denote the \(i\)-th synthetic channel of the length-\(N\) polar transform.  The Bhattacharyya parameter of \(W_N^{(i)}\) is
\(
Z\big(W_N^{(i)}\big)
\;\triangleq\;
\sum_{y\in\mathcal{Y}^N}\sqrt{\,W_N^{(i)}(y\mid 0)\;W_N^{(i)}(y\mid 1)\,}\,,
\)
which measures the reliability of the \(i\)-th synthetic channel (smaller \(Z\) means more reliable).
\end{definition}

\begin{definition} [Information and frozen sets]
Let $ \mathbf{A}\subseteq[1:N]$ denote the \emph{information set} and let
\(
 \mathbf{F} \;=\; [1:N]\setminus  \mathbf{A}
\)
denote the \emph{frozen set}.
The information set \( \mathbf{A}\) is formed from the synthetic-channel
reliabilities \(\{W_N^{(i)}\}_{i=1}^N\). The common selection rule is:
\(
 \mathbf{A} \;=\; \{\, i\in[1:N] : Z(W_N^{(i)}) \le \zeta_N \,\}
\)
where \(0 \le \zeta_N <1\) are design thresholds that vanish appropriately
as \(N\to\infty\). 
\end{definition}

\begin{remark} \label{rem:selection}
For a target rate \(R\) and blocklength \(N\) we choose the information set \(|\mathbf{A}|=\lfloor N R\rfloor\) (using the selection rule described above). The frozen set is \(\mathbf{F}=[1:N]\setminus\mathbf{A}\). These sets \(\mathbf{A}\) and \(\mathbf{F}\) are fixed for the remainder of the paper.
\end{remark}

\subsection{Motivation: Polar coding with non-identical BEC uses and time-shared physical channels}
Consider the simplest nontrivial polar encoder with blocklength \(N=2\) and generator
\(
G_2\),
\(
\mathbf{u}=(u_F,u_A),
\)
where \(u_A\) is an information bit and \(u_F\) is a frozen bit. The encoded codeword is
\(
\mathbf{x}=\mathbf{u}G_2=(x_1,x_2)=(u_A\oplus u_F,\;u_A).
\)
Suppose the public symbol is \(x_1=u_A\oplus u_F\). If the frozen bit \(u_F\) is drawn uniformly at random (i.i.d.\ \(\mathrm{Bernoulli}(1/2)\)) and kept secret between sender and legitimate receiver, then \(\Pr(u_A\mid x_1)=\tfrac12\), so \(x_1\) alone reveals no information about \(u_A\). In this toy example the pair \((x_1,u_F)\) functions as a one-time pad: the legitimate receiver recovers \(u_A=x_1\oplus u_F\), while an adversary observing only \(x_1\) learns nothing about \(u_A\). This observation motivates constructing an information-theoretic symmetric cryptosystem based on using frozen bits as encryption keys in practice, but only so long as the masks are fresh and never reused across messages. Conversely, publishing the coordinate that directly carries the information (here \(x_2=u_A\)) or any symbol that leaks the mask would trivially compromise secrecy, illustrating that which coordinate is made public matters even in the smallest code. Hence, in the finite-length regime we require a concrete, efficiently computable leakage measure that decides which indices are safe to expose; this need motivates a practical bit selection and computing the amount of the leakage, which we will talk about it later. 

We consider polar coding when different channel uses experience different erasure probabilities and the physical link is time-shared between a \emph{public} and a \emph{private} channel. Concretely, let \(N=2^n\) be a polar blocklength and for \(i=1,\dots,N\) let \(\delta_i\in[0,1]\) denote the erasure probability of the channel used for the \(i\)-th transmitted bit (i.e., the \(i\)-th physical use is a \(\mathrm{BEC}(\delta_i)\)). Collect these into the vector \(\boldsymbol{\delta}=(\delta_1,\dots,\delta_N)\). We assume two physical channels with fixed qualities \(\delta_{\mathrm{pub}}\) (public) and \(\delta_{\mathrm{priv}}\) (private), and model scheduling/time-sharing by assigning each position \(i\) to one of these channels so that \(\delta_i\in\{\delta_{\mathrm{pub}},\delta_{\mathrm{priv}}\}\). Time can be partitioned into intervals, \(T_1,T_2,\dots\), and a transmitter is mixing uses of the public and private links across intervals; in other words, time-sharing is natural and trivial to implement in practice. The toy \(N=2\) example illustrates this operationally: one can send \(x_1=u_A\oplus u_F\) over the public link while sending \(x_2=u_A\) over the private link, so that the public exposure is masked by the frozen bit, but the choice of which coordinate is public matters and must be analyzed carefully.

The leakage identity holds for any binary linear code, but we focus on polar codes because their frozen-bit structure naturally provides masking randomness and their recursive Kronecker form supports efficient construction and certification. We interpret the frozen bits as fresh random masks, so the scheme can be viewed as a polar-based coset code rather than standard frozen-to-zero polar coding. Unlike a one-time pad, our goal is not direct encryption, but exact finite-length leakage control when selected codeword coordinates are intentionally exposed on a public resource. The BEC model serves only as a motivating operational example for public/private scheduling and does not affect the rank-based leakage identity.


\section{RankGuard-Polar: sending part of the codeword publicly} \label{sec:RankGuard}

\subsection{Assumption for RankGuard-Polar }

\textbf{\emph{RankGuard--Polar} encoder:} The RankGuard--polar construction addresses the setting in which the eavesdropper is assumed to have access to the channel input, reflecting the public nature of the channel, and hence must be secured even under this strong eavesdropper assumption.

\textbf{Information and Frozen indices:} As stated in Remark~\ref{rem:selection}, once the information and frozen sets \(\mathbf{A}\) and \(\mathbf{F}\) are selected for a target rate \(R\), they are held fixed for the remainder of the paper. In particular, for the \emph{RankGuard--Polar} construction considered here we employ the same fixed sets \(\mathbf{A}\) and \(\mathbf{F}\) across all channel models studied. We do \emph{not} re-select or adapt \(\mathbf{A}\) and \(\mathbf{F}\) separately for different channels. This choice isolates the effect of the RankGuard mechanism and simplifies the ensuing rank and secrecy analyses.

\textbf{\emph{RankGuard--Polar} decoder:} We also assume the decoder for the \emph{RankGuard--Polar} construction is identical to the standard polar-code decoder, the successive-cancellation (SC) decoder of \cite{Arikan2009}.

Sending raw codeword symbols through a public channel is, in general, incompatible with privacy preservation: an adversary who observes those symbols may learn information about the message. At the same time, many practical systems possess public resources (broadcast links, side-channels, cloud storage) that can be exploited to increase throughput or reliability, if used carefully.

We adopt the following model. Let \(N=2^n\) and let \(G_N\) be Arıkan's polar transform. The encoder forms \(\mathbf{x} = \mathbf{u} G_N\), where \(\mathbf{u}=(\mathbf{u_A},\mathbf{u_F})\) partitions the vector of bits into information bits \(\mathbf{u_A}\) (indices in the information set \(\mathbf{A}\)) and frozen bits \(\mathbf{u_F}\) (indices in the frozen set \(\mathbf{F}\)). We assume the frozen bits \(\mathbf{u_F}\) are drawn i.i.d.\ \(\mathrm{Bernoulli}(1/2)\) and are known to the legitimate receiver but not to external adversaries. This is a crucial assumption: frozen bits must act as a shared secret key between encoder and decoder. The channel inputs \(x_1,\dots,x_N\) are produced by the standard linear mapping \(\mathbf{x}=\mathbf{u}G_N\).
Before relating leakage to mutual information, it is helpful to view the linear maps separately: We propose to send a carefully chosen subset of coordinates \(\mathbf{P}\subset[1:N]\) over a public channel (the \emph{public subchannel}), and send the remaining coordinates over private channels. The goal is that the public outputs \(\mathbf{x_p}\) are (almost) independent of the message bits \(\mathbf{u_A}\), i.e., they reveal negligible information about \(\mathbf{u_A}\).

\subsection{Theorems and Examples for RankGuard-Polar}

\begin{lemma} \label{lem:cond_entropy}
Let \(\mathbf{u}=(\mathbf{u_A},\mathbf{u_F})\) with \(\mathbf{u_F}\) uniform on \(\mathbb F_2^{|F|}\) and independent of \(\mathbf{u_A}\).
For any published index set \(\mathbf{P}\) and \(\mathbf{x_p} = \mathbf{u_A}\, G_{\mathbf{A,P}} \oplus \mathbf{u_F} \, G_{\mathbf{F,P}}.
\) we have
\begin{align}
H(\mathbf{x_p}\mid \mathbf{u_A})=\operatorname{rank}\big(G_{\mathbf{F,P}}\big).
\end{align}
\end{lemma}

\textbf{Proof:}
Fix \(\mathbf{u_A}\) and set \(\mathbf{A}\triangleq \mathbf{u_A}G_{\mathbf{A,P}}\). Then
\(
H(\mathbf{x_p}\mid \mathbf{u_A})
=H\big(a\oplus \mathbf{u_F}G_{\mathbf{F,P}}\mid \mathbf{u_A}\big)
=H\big(\mathbf{u_F}G_{\mathbf{F,P}}\mid \mathbf{u_A}\big),
\)
since adding the deterministic vector \(\mathbf{A}\) does not change entropy. As \(\mathbf{u_F}\) is independent of \(\mathbf{u_A}\) and uniform on \(\mathbb F_2^{|F|}\),
\begin{align}
H\big(\mathbf{u_F}G_{\mathbf{F,P}}\mid \mathbf{u_A}\big)=H\big(\mathbf{u_F}G_{\mathbf{F,P}}\big).
\end{align}
The linear map \(G_{\mathbf{F,P}}:\mathbb F_2^{|F|}\to\mathbb F_2^{|P|}\) has image \(\mathbf{V}\) with \(|V|=2^{r}\), where \(r=\operatorname{rank}(G_{\mathbf{F,P}})\); the pushforward of the uniform distribution on \(\mathbb F_2^{|F|}\) is uniform on \(\mathbf{V}\). Hence
\(
H\big(G_{\mathbf{F,P}}\mathbf{u_F}\big)=\log_2|V|=\log_2 2^{r}=r,
\)
so \(H(\mathbf{x_p}\mid \mathbf{u_A})=\operatorname{rank}(G_{\mathbf{F,P}})\), as required.
$\Box$

\begin{theorem}[Leakage certificate]
Let \(\mathbf{A,F}\) be disjoint index sets, \(\mathbf{A} \bigcup \mathbf{F}=[1:N]\), and let \(\mathbf{P}\subset[1:N]\). Define:
\begin{align}
\mathbf{u} \triangleq  \begin{bmatrix}\mathbf{u_A} & \mathbf{u_F}\end{bmatrix}
\in\mathbb F_2^{1\times(|A|+|F|)},\qquad
\mathbf{x_p}=\mathbf{u}\,G_{\mathbf{P}},
\end{align}
\begin{align}
G_{\mathbf{P}}\triangleq \begin{bmatrix} G_{\mathbf{A,P}} \\[4pt] G_{\mathbf{F,P}} \end{bmatrix}
\in\mathbb F_2^{(|A|+|F|)\times |P|},
\end{align}

Since \(\mathbf{u_F}\) is independent of \(\mathbf{u_A}\) and uniform on \(\mathbb F_2^{|F|}\) the following bounds hold:
\(
H(\mathbf{x_p})=\operatorname{rank}(G_{\mathbf{P}})
\),
\(H(\mathbf{x_p}\mid \mathbf{u_A})=\operatorname{rank}\big(G_{\mathbf{F,P}}\big),
\)
and therefore the leakage certificate \(L(\mathbf{P})\) is:
\begin{align}
L(\mathbf{P})\triangleq I(\mathbf{u_A};\mathbf{x_p})=
\operatorname{rank}\big(G_{\mathbf{P}}\big)-\operatorname{rank}\big(G_{\mathbf{F,P}}\big).
\label{eq:LP}
\end{align}
\end{theorem}

\textbf{Proof:}
 Since \(\mathbf{u}\) is uniform on \(\mathbb F_2^{|A|+|F|}\), equivalently \(\mathbf{u_A}\) and \(\mathbf{u_F}\) are independent and each uniform. The condition that \(\mathbf{u}\) is uniform on \(\mathbb F_2^{|A|+|F|}\) implies that \( \mathbf{x_p}=\mathbf{u}G_{\mathbf{P}}\) is the pushforward of a uniform distribution and therefore is uniform on \(\mathrm{Im}(G_{\mathbf{P}})\), whence \(H(\mathbf{x_p})=\log_2|\mathrm{Im}(G_{\mathbf{P}})|=\operatorname{rank}(G_{\mathbf{P}})\). Also, the conditional equality
\(H(\mathbf{x_p}\mid \mathbf{u_A})=\operatorname{rank}(G_{\mathbf{F,P}})\) is exactly the content of Lemma~\ref{lem:cond_entropy}. $\Box$

The $L(\mathbf{P})$ equals the number of independent linear combinations of the information vector $\mathbf{u_A}$ that an adversary can recover from the public observation $\mathbf{x_p}$. In particular, $L(\mathbf{P})=0$ implies no information about $\mathbf{u_A}$ is revealed, which means that \(\operatorname{rank}\big(G_{\mathbf{P}}\big)=\operatorname{rank}\big(G_{\mathbf{F,P}}\big)\) and all columns of $G_{\mathbf{A,P}}$ lie in the column space of $G_{\mathbf{F,P}}$, while $L(\mathbf{P})=r>0$ means up to $r$ independent bits of $\mathbf{u_A}$ are exposed.

\textbf{Example 1:} Consider the \(N=4\) polarized codeword whose coordinates satisfy:
\begin{align}
x_1 = u_4 \oplus u_3 \oplus u_2 \oplus u_1, \quad x_2 = u_4 \oplus u_2,\\
x_3 = u_4 \oplus u_3, \quad x_4 = u_4,
\label{eq:N=4}
\end{align}
where \(u_4,u_3,u_2,u_1\in\mathbb F_2\).  With \(\mathbf{u}=[u_1\;u_2\;u_3\;u_4]\) we may write \(\mathbf{x}=\mathbf{u}G=[x_1\;x_2\;x_3\;x_4]\) with
\begin{align}
G \;=\;
\begin{bmatrix}
1 & 0 & 0 & 0\\[4pt]
1 & 1 & 0 & 0\\[4pt]
1 & 0 & 1 & 0\\[4pt]
1 & 1 & 1 & 1
\end{bmatrix}.
\end{align}
\textbf{Scenario 1:} information set is \(\mathbf{A}=\{4\}\), means \(u_4\), and the frozen set is \(\mathbf{F}=\{1,2,3\}\) means \(\{u_3,u_2,u_1\}\).
\begin{itemize}
  \item Now, let us make \(x_4\) public, so we have \(\mathbf{P}=\{4\}\). Now, publishing \(x_4\) reveals exactly one bit of \(\mathbf{u_A}\), so the adversary recovers \(u_4\) completely. The mapping matrices are
\(
G_{\mathbf{A,P}}=[\,1\,],\qquad G_{\mathbf{F,P}}=[\,0\;0\;0\,]^T.
\)
Hence \(\operatorname{rank}(G_{\mathbf{P}})=1\), \(\operatorname{rank}(G_{\mathbf{F,P}})=0\) and
\(
L(\{4\})=1-0=1.
\)
  \item Let us make \(x_1\) public, so we have \(\mathbf{P}=\{1\}\). The single published symbol \(x_1=u_4 \oplus u_3 \oplus u_2 \oplus u_1\) carries no information about \(u_4\).  The mapping matrices are
\(
G_{\mathbf{P}}=[\,1\,1\;1\;1\,]^T,\qquad G_{\mathbf{F,P}}=[\,1\;1\;1\,]^T.
\)
Here \(\operatorname{rank}(G_{\mathbf{P}})=1\) and \(\operatorname{rank}(G_{\mathbf{F,P}})=1\), so
\(
L(\{1\})=1-1=0.
\)
Hence \(x_1\) alone is masked by \(u_3,u_2,u_1\) and reveals no information about \(u_4\).
\end{itemize}

\textbf{Scenario 2:} information set is \(\mathbf{A}=\{4,3,2\}\), , and the frozen set is \(\mathbf{F}=\{1\}\).
\begin{itemize}
  \item If \(\mathbf{P}=\{1\}\),  then
\(
G_{\mathbf{P}}=[\;1\;1\;1\;1\,]^T,\qquad G_{\mathbf{F,P}}=[\,1\,],
\)
and \(\operatorname{rank}(G_{\mathbf{P}})=1\), \(\operatorname{rank}(G_{\mathbf{F,P}})=1\), so
\(
L(\{1\})=1-1=0.
\)
Thus \(x_1\) alone does not leak information about \((u_4,u_3,u_2)\).
\(x_1\) is masked by \(u_1\) and alone reveals no information about the triple \((u_4,u_3,u_2)\).

\item If \(\mathbf{P}=\{1,2,3\}\), then the mapping submatrices are
\begin{align}
G_{\mathbf{P}}
=\begin{bmatrix}1&0&0\\[2pt]1&1&0\\[2pt]1&0&1\\[2pt]1&1&1\end{bmatrix},\qquad
G_{\mathbf{F,P}}=[\;1\;0\;0\,].
\end{align}
Row-reduction yields \(\operatorname{rank}(G_{\mathbf{P}})=3\) and \(\operatorname{rank}(G_{\mathbf{F,P}})=1\), hence
\(
L(\{1,2,3\}) = 3-1 = 2.
\)
Thus, up to two independent linear relations on the information triple \((u_4,u_3,u_2)\) can be recovered by observing \((x_1,x_2,x_3)\). 
In fact, by having \((x_1,x_2,x_3)\) being public, although \(x_1\) does not leak anything about \(u_4,u_3,u_2\), since \(u_1\) is masked the leakage, \( x_2 = u_4\oplus u_2\), \(x_3 = u_4\oplus u_3,\) gives two equations that helps attacker to learn the difference of the information.
\end{itemize}

\begin{theorem}\label{thm:leak-matrix-R}
Fix disjoint index sets $\mathbf{A,F}$, \(\mathbf{A \bigcup F}=[1:N]\), and let $\mathbf{P}\subset[1:N]$. Assume \eqref{eq:xp}, let
\(\mathbf{Q} \triangleq  \operatorname{rowspace}(G_{\mathbf{P}})\subseteq\mathbb F_2^{1\times|P|}\),
\(\mathbf{V} \triangleq  \operatorname{rowspace}(G_{\mathbf{F,P}})\subseteq \mathbf{Q},\)
and define
\(
r \triangleq  \dim \mathbf{Q} - \dim \mathbf{V}
= \operatorname{rank}(G_{\mathbf{P}})-\operatorname{rank}(G_{\mathbf{F,P}}).
\)
Then there exists a matrix $R\in\mathbb F_2^{|P|\times r}$ such that
\begin{align}
G_{\mathbf{F,P}}\,R = 0 \;\text{and}\; \operatorname{rank}(G_{\mathbf{A,P}}\,R)=r>0.
\end{align}
Consequently the $r$ columns of $R$ yield $r$ independent linear equations
in $\mathbf{u_A}$ obtained by right-multiplying the observation:
\begin{align}
\mathbf{x_p}\,R \;=\; \mathbf{u_A}\,(G_{\mathbf{A,P}}R) \in\mathbb F_2^{1\times r}.
\end{align}
\end{theorem}

\textbf{Proof:}
Because \(\mathbf{V}\subseteq \mathbf{Q}\) and \(\dim(\mathbf{Q}/V)=r\), there exist vectors
\(\bar q_1,\dots,\bar q_r\in \mathbf{Q}\) whose cosets \(\bar q_1+\mathbf{V},\dots,\bar q_r+\mathbf{V}\) form a basis of the quotient space \(\mathbf{Q}/\mathbf{V}\).
Equivalently, if \(v_1,\dots,v_{d_V}\) is a basis of \(\mathbf{V}\), then
\(
\{v_1,\dots,v_{d_V},\bar q_1,\dots,\bar q_r\}
\)
is a basis of \(\mathbf{Q}\). Define a linear map \(\phi:\mathbf{Q}\to\mathbb F_2^{1\times r}\) on this basis by
\(
\phi(v_j)=0\quad(j=[1:d_V]),\,\,
\phi(\bar q_\ell)=e_\ell^\top\quad(\ell=[1;r]),
\)
where \(e_\ell^\top\) is the \(\ell\)-th standard row vector in \(\mathbb F_2^{1\times r}\).
By linearity this uniquely determines \(\phi\) on all of \(\mathbf{Q}\).  Because \(\phi\) sends the
\(\bar q_\ell\) to the standard basis of \(\mathbb F_2^{1\times r}\), the image \(\phi(\mathbf{Q})\)
contains the \(r\) independent rows \(e_1^\top,\dots,e_r^\top\); therefore
\(\dim \phi(\mathbf{Q})=r\), i.e. \(\operatorname{rank}(\phi)=r\), \cite{AxlerLADR}.

Next extend \(\phi\) arbitrarily to a linear map \(\widetilde \phi:\mathbb F_2^{1\times|P|}\to\mathbb F_2^{1\times r}\).
Any linear map from row-vectors of length \(|\mathbf{P}|\) to row-vectors of length \(r\) is realized
by right-multiplication with a matrix \(R\in\mathbb F_2^{|P|\times r}\); that is,
\(\widetilde \phi(y)=yR\) for all \(y\in\mathbb F_2^{1\times|P|}\).  Because \(\phi\) vanishes on \(\mathbf{V}\),
the extension \(\widetilde \phi\) also vanishes on \(\mathbf{V}\), hence every row of \(G_{\mathbf{F,P}}\) is
mapped to zero and we have \(G_{\mathbf{F,P}}R=0\).

Now observe that every row of \(G_{\mathbf{A,P}}\) lies in \(\mathbf{Q}\). Applying \(\widetilde \phi\) row-wise
to \(G_{\mathbf{A,P}}\) is therefore the same as applying \(\phi\) to the rowspace \(\mathbf{Q}\); in matrix form
this is exactly the product \(G_{\mathbf{A,P}}R\). Consequently the rowspace of \(G_{\mathbf{A,P}}R\) equals
\(\phi(\mathbf{Q})\), which has dimension \(r\). Hence
\(
\operatorname{rank}(G_{\mathbf{A,P}}R)=\dim \phi(\mathbf{Q})=r.
\)
Finally, right-multiplying the public observation \(\mathbf{x_p}=\mathbf{u_A}G_{\mathbf{A,P}}\oplus\mathbf{u_F}G_{\mathbf{F,P}}\)
by \(R\) gives
\(
\mathbf{x_p}R=\mathbf{u_A}(G_{\mathbf{A,P}}R)\oplus\mathbf{u_F}(G_{\mathbf{F,P}}R)=\mathbf{u_A}(G_{\mathbf{A,P}}R),
\)
since \(G_{\mathbf{F,P}}R=0\). The \(r\) columns of \(R\) therefore yield \(r\) independent linear
equations in \(\mathbf{u_A}\), completing the proof. \(\Box\)

Now, let's get back to our example with two scenarios:

\textbf{Example 2:}
Consider $N=4$, and scenarios of Example 1.

\textbf{Scenario 1.} $\mathbf{A}=\{u_4\}$ and $\mathbf{F}=\{u_3,u_2,u_1\}$.

\begin{itemize}
\item $\mathbf{P}=\{4\}$ (publish $x_4$). Since \(L(\mathbf{P})=1\),
choose
\(
R=\begin{bmatrix}1\end{bmatrix}\in\mathbb F_2^{1\times 1},
\)
then
\(
G_{\mathbf{F,P}}\,R=0,
G_{\mathbf{A,P}}\,R=\begin{bmatrix}1\end{bmatrix}\begin{bmatrix}1\end{bmatrix}=\begin{bmatrix}1\end{bmatrix},
\)
and $\operatorname{rank}(G_{\mathbf{A,P}}R)=1$. The single leaked equation is
\(x_4 = u_4. \)

\item $\mathbf{P}=\{1\}$ (publish $x_1$), since \(L(\mathbf{P})=0\), hence no nontrivial $R\in\mathbb F_2^{1\times 0}$ exists and $x_1$ alone reveals no information about $u_4$.
\end{itemize}

\textbf{Scenario 2.} $\mathbf{A}=\{u_4,u_3,u_2\}$ and $\mathbf{F}=\{u_1\}$.

\begin{itemize}
\item $\mathbf{P}=\{1\}$ (publish $x_1$). Since \(L(\mathbf{P})=0\), thus, $x_1$ alone does not leak information about $(u_4,u_3,u_2)$ and no nontrivial $R$ exists.

\item $\mathbf{P}=\{1,2,3\}$ (publish $x_1,x_2,x_3$). Since \(L(\mathbf{P})=2\), one convenient choice is
\begin{align}
R =
\begin{bmatrix}
0 & 0 \\[2pt]
1 & 0 \\[2pt]
0 & 1
\end{bmatrix}\in\mathbb F_2^{3\times 2}.
\end{align}
Check:
\(
G_{\mathbf{F,P}}\,R = [0\;0],
\)
and
\begin{align}
G_{\mathbf{A,P}}\,R =
\begin{bmatrix}
1 & 1\\[2pt]
0 & 1\\[2pt]
1 & 0
\end{bmatrix}\in\mathbb F_2^{3\times 2},
\end{align}
which has rank $2$. Right-multiplying the observation yields
\(
\mathbf{x_p} R = \mathbf{u_A} (G_{\mathbf{A,P}} R).
\)
With $\mathbf{u_A}=[\,u_4\; u_3\; u_2\,]$ and $\mathbf{x_p}=[\,x_1\; x_2\; x_3\,]$ we get
\(
\mathbf{x_p} R = \begin{bmatrix}x_2 & x_3\end{bmatrix},\)
\(\mathbf{u_A} (G_{\mathbf{A,P}}R) = \begin{bmatrix}u_4\oplus u_2 & u_4\oplus u_3\end{bmatrix},
\)
so the two independent leaked equations are
\(
x_2 = u_4 \oplus u_2, x_3 = u_4 \oplus u_3.
\)
Any other $R$ whose first row is zero and for which $G_{\mathbf{A,P}}R$ has rank $2$ is equally valid.
\end{itemize}

\section{Score Greedy Algorithm} \label{sec:Algorithms}
Computing the exact leakage \(L(\mathbf{P})\) for every candidate \(k\)-set \(\mathbf{P}\) is combinatorially expensive (different \(\mathbf{P}\) produce different submatrices). In fact, exact minimization of \(L(\mathbf{P})\) by exhaustive search is combinatorial: there are \(\binom{N}{k}\) candidate \(k\)-sets, and evaluating one candidate requires computing the rank of an \(N\times k\) submatrix, which costs \(O(Nk^2)\)
(for \(k\le N\)) using Gaussian elimination. 
Hence, a brute-force (BF) solver requires \(T_{BF}=\binom{N}{k}\cdot O(Nk^2)\) time and is impractical except for very
small \(N,k\). For this reason, we use \textsc{ScoreGreedy} (SG) as a fast initializer so we adopt a lightweight row-based surrogate and search over rows instead of enumerating all \(\mathbf{P}\). Algorithm \textsc{ScoreGreedy} scores each row \(i\) by
\(
a_i\triangleq \|G^{(i)}_{P,A}\|_0,\) \(f _i\triangleq \|G^{(i)}_{P,F}\|_0\),\(s_i\triangleq f_i- a_i,
\)
selects the top-\(k\) rows by \(s_i\), and returns their index set \(\mathbf{P}\). This is fast (cost dominated by \(O(N^2)\) computing the counts plus \(O(N\log N)\) to sort, which gives totally \(T_{SC}= O(N^2)\) time) and intentionally favors rows that touch many frozen columns but few information columns.

The following theorem, Theorem~\ref{thm:combined}, justifies this surrogate choice: it gives a deterministic bound showing that controlling the per-row counts $a_i$ yields an upper bound on the leakage. Thus, minimizing the sum $\sum_{i\in\mathbf{P}} a_i$ over $k$-sets is a natural surrogate objective for leakage minimization.

\begin{theorem}[ScoreGreedy leakage control]
\label{thm:combined}
Let \(G\in\mathbb F_2^{N\times N}\) be the generator matrix and let the column partition be \(\mathbf{A}\) (information) and \(\mathbf{F}\) (frozen). For each row \(i\in[1:N]\). Let \(1 \le k \le |\mathbf{F}|\) and it denotes that we choose the \(k\)-set returned by ScoreGreedy (top-\(k\) rows by \(s_i\)). For every \(\mathbf{P}\subset[1:N]\),
  \begin{align}
  L(\mathbf{P})\le \operatorname{rank}(G_{\mathbf{A,P}}) \le \sum_{i\in \mathbf{P}, |\mathbf{P}|=k} a_i.
  \end{align}
\end{theorem}

\textbf{Proof:}
Subadditivity of rank gives \(\operatorname{rank}(G_{\mathbf{P}})\le\operatorname{rank}(G_{\mathbf{A,P}})+\operatorname{rank}(G_{\mathbf{F,P}})\). Since \(L(\mathbf{P})=\operatorname{rank}(G_{\mathbf{P}})-\operatorname{rank}(G_{\mathbf{F,P}})\), we will have \(L(\mathbf{P}) \le \operatorname{rank}(G_{\mathbf{A,P}})\). The rank of \(G_{\mathbf{A,P}}\) is at most the number of its nonzero columns, which in turn is \(\le\sum_{i\in \mathbf{P}}a_i\) because each nonzero column contributes at least one 1 among the rows \(\mathbf{P}\). $\Box$

\begin{algorithm}
  \caption{ScoreGreedy Algorithm}
  \label{alg:ScoreGreedy}
  \begin{algorithmic}[1]
    \Require $G_N$, $A$, $F$, budget $k$
    \Ensure public set $P$
    \For{$i\gets 1$ to $N$}
      \State $a_i\gets \|G^{(i)}_{P,A}\|_0,\quad f_i\gets\|G^{(i)}_{P,F}\|_0$
      \State $s_i\gets f_i- a_i$
    \EndFor
    \State $P\gets$ top-$k$ indices by $s_i$ (tie-break as desired)
    \State \Return $P$
  \end{algorithmic}
\end{algorithm}


\section{Relations to previous work} \label{sec:Relations}

Our problem departs from the canonical wiretap formulation in several crucial respects. Concretely, we intentionally exploit a publicly available resource and publish a chosen subset of codeword coordinates over a public subchannel; an eavesdropper may observe the channel \emph{inputs} corresponding to those published coordinates and attempt to recover parts of the message. By contrast, the classical wiretap literature typically assumes the eavesdropper observes the \emph{output} of an adversarial channel and asks which secrecy rates are achievable, constructing codes so that the eavesdropper's information vanishes asymptotically as the blocklength \(N\to\infty\) \cite{wyner1975,Wei2015,BroadcastHassani}. The strong, operationally different assumption that the adversary may see the channel \emph{input} on a public resource allows our method to treat an eavesdropper's channel as a public channel and deliberately publish selected codeword coordinates on it; the analysis then determines how many bits (and which linear combinations) are exposed when those coordinates are placed on the public resource. Accordingly, rather than providing asymptotic secrecy-rate guarantees, our approach yields exact finite-\(N\) leakage measures and constructive selection algorithms (algebraic/rank certificates) tailored to the published-index setting.

In this work, at finite blocklength, we quantify the leakage induced by a published index set \(\mathbf{P}\) when the eavesdropper observes the public-channel input. We provide exact rank-based leakage certificates and constructive descriptions of the exposed linear combinations.

Moreover, our model departs from common cryptographic assumptions that consider a single public channel available to an adversary: we explicitly model a pair of physical links (private and public) and allow the encoder to schedule different codeword coordinates across them. In other words, the transmitter will deliberately send some coordinates over a protected private link while publishing others on a public link, and leakage must be assessed for this mixed public/private usage rather than for a single-channel wiretap scenario.

Unlike the classical wiretap-II model, \cite{ozarow1984wiretap2}, which studies secrecy against observation of a subset of codeword symbols, our framework treats the public symbols as intentionally exposed and quantifies their exact finite-length leakage under a polar-based masking model.


\section{Conclusion} \label{sec:Conclusion}
We presented RankGuard-Polar, a finite-block toolkit for publishing polar-code coordinates while certifying exact information leakage. The key technical insight is the rank-difference identity linking mutual information to submatrix ranks, which enables both constructive extractors and fast GF(2) certification procedures. RankGuard-Polar is intended as an operational audit-and-control layer: it does not replace asymptotic wiretap constructions but provides finite-\(N\) guarantees and post-hoc certificates that are useful in real systems that must exploit public resources.

\bibliographystyle{unsrt}
\bibliography{bibD}

\end{document}